\documentclass[12pt]{amsart}
\usepackage[USenglish]{babel}
\usepackage{amsaddr}
\usepackage[numbers]{natbib}

\usepackage[letterpaper,top=2cm,bottom=2cm,left=3cm,right=3cm,marginparwidth=1.75cm]{geometry}

\usepackage{amsmath}
\usepackage{bm}
\usepackage{graphicx}
\usepackage[colorlinks=true, allcolors=blue]{hyperref}

\usepackage{amsfonts,amssymb,amsthm,amsmath}
\usepackage{graphicx}
\usepackage{float}
\usepackage{comment}

\usepackage{pgf,tikz, tikz-3dplot}
\usepackage{pgfplots}

\usetikzlibrary{arrows}
\usetikzlibrary{patterns,intersections}
\usetikzlibrary{math}
\usetikzlibrary{decorations.pathreplacing,angles,quotes}
\usetikzlibrary{decorations.markings}
\usetikzlibrary{positioning}
\usetikzlibrary{through}
\usetikzlibrary{tikzmark}
\usetikzlibrary{angles,quotes}
\usetikzlibrary{calc}
\pgfplotsset{compat=1.10}
\pgfplotsset{compat=1.17}
\usepgfplotslibrary{fillbetween}

\numberwithin{equation}{section}
\theoremstyle{plain}
\newtheorem{theorem}{Theorem}[section]
\newtheorem{proposition}[theorem]{Proposition}

\theoremstyle{definition}
\newtheorem{definition}[theorem]{Definition}
\newtheorem{example}[theorem]{Example}

\newtheorem*{axiom*}{Axiom}
\newtheorem{axiom}{Axiom}
\theoremstyle{remark}
\newtheorem{remark}[theorem]{Remark}

\newcommand{\R}{\mathbb{R}}
\newcommand{\N}{\mathbb{N}}

\title{Revisiting the Measurement of Polarization}
\author{Juan A. Crespo}
\address{Departamento de Economía Cuantitativa, Universidad Aut\'onoma de Madrid}
\author{Armajac Raventós-Pujol}
\address{Departamento de Econom\'ia, Universidad Carlos III de Madrid }
\email{juan.crespo@uam.es, aravento@eco.uc3m.es }


\thanks{Both authors are partially supported by Comunidad de Madrid (Spain) excelence grant H2019/HUM-5891 and Spanish MICIU  grant PID2024-161056NB-I00. Second author is partially supported by grants of the Spanish MINECO  CEX2021-001181-M and MICIU/AEI /10.13039/501100011033. They are deeply indebted to Ignacio Ortuño and Klaus Desmet for their insightful questions and comments which greatly improved this paper. They also thank Shlomo Weber, Juan de Dios Moreno and Antonio Cabrales for their helpful suggestions. The paper was presented at the 13th Conference on Economic Design, the 48th Simposio de Análisis Económico, and in seminars at Universidad Carlos III and Universidad Autónoma de Madrid. The authors are grateful to all attendees for their valuable feedback.}

\begin{document}

\begin{abstract}
We revisit Esteban and Ray’s (1994) seminal model of polarization. Their main result (unnecessarily) relies on the assumption that individuals are infinitely divisible, which imposes strong restrictions on admissible polarization indices. We show that relaxing this assumption yields a broader family of indices consistent with the original axioms. The resulting indices avoid counter-intuitive rankings that arise when using results on the original paper and provide greater flexibility for empirical applications.
\end{abstract}
	\maketitle


\section*{Introduction}

It has been 30 years since Esteban and Ray’s (1994) seminal contribution to the measurement of polarization~\cite{10.2307/2951734}. In their model, antagonism in society depends on the alienation individuals feel toward members of other groups, as well as on their identification with their own group. They introduce a set of axioms that formalize the concept of polarization, thereby narrowing the set of admissible measures interpretable as polarization indices. The Esteban–Ray measure has since become well established and widely applied in economics and political science to address diverse empirical and theoretical questions~\cite{Aghion2004, Alesina2003, https://doi.org/10.3982/ECTA13734, 10.1257/aer.20180174, 10.1257/app.20190309, https://doi.org/10.3982/ECTA11237, 10.1257/00028280260344722, 10.1257/aer.20160812}.\\

The present paper makes three contributions. First, using the original model and axiomatization, we show that the concept of polarization admits a broader set of measures with fewer parameter restrictions. By considering that individuals are infinitely divisible, Esteban and Ray (1994) make a simplifying assumption that is not innocuous: it sharply limits the set of feasible polarization indices. Second, we provide examples in which these restrictive measures yield counterintuitive rankings of societies. In these cases, a more flexible antagonism function, which is consistent with the original model and axiomatization, produces rankings that align better with intuitive judgments. Third, we show that expanding the feasible set of indices gives empirical researchers greater flexibility. Depending on the context, a researcher may wish to adopt different functional forms or parameterization for how antagonism increases with group size or with inter-group distance.\\

To identify the broader set of polarization indices consistent with the Esteban–Ray model and axiomatization, we begin by removing the assumption that individuals are infinitely divisible. We then re-examine the role of Condition H and Axioms 1–3. We show that Condition H is equivalent to requiring that the antagonism function decomposes as the product of identification and alienation, where identification is a non-negative power function of group size. We then prove that Axiom 1 is equivalent to requiring that \emph{the role of identification cannot be neutral}. Axiom 2 is equivalent to assuming that \emph{the marginal effect of distance on alienation is non-decreasing}. Finally, while Axiom 1 concerns identification and Axiom 2 concerns alienation, Axiom 3 rules the relationship between the two. More specifically, once an alienation function has been set, Axiom 3 restricts the feasible range of powers for the identification function and vice versa.
These three equivalences allow us to view the axioms from a new perspective: whereas Esteban and Ray formulated them using a \emph{macro perspective}, namely, how polarization depends on the shape of social distributions, they can now be reformulated through a \emph{micro perspective}, namely, the properties of individual antagonism functions.

Not only that, when one does not admit that  individuals are infinitely divisible. the precise characterization we obtain of antagonism functions enables us to construct new examples  of indices (see Section~\ref{new}) that satisfy Condition H and the axioms.    Those measures of polarization cannot be extended to the domains for populations used by ER, namely the positive real numbers. Thus, these examples show that the simplification made in \cite{10.2307/2951734} is the reason  they obtained such a restrictive family of polarization indices\footnote{ In words by ER: \emph{``This theorem dramatically narrows the class of allowable polarization measures."}} and, therefore, that it was a technical assumption that was not harmless.\\

Our paper is not a theoretical quibble about a technical assumption, we show the restricted set of indices obtained by Esteban and Ray leads to counter-intuitive results. Specifically, we compare distributions of two societies that are easily rankable based on our intuitive understanding of which is more polarized. We then show that Esteban–Ray indices would rank these societies erroneously, whereas the broader set of indices we identify would allow for correct rankings (see Section~\ref{expl} for details).\\

In that sense, our findings are important for empirical work, as they allow for a broader range of possibilities when modeling the effects of group size and intergroup distance on antagonism. For instance, we suggest revisiting articles that employed the Esteban–Ray index and dismissed the role of polarization in certain phenomena (e.g., the effect of political polarization on government size~\cite{LINDQVIST_OSTLING_2010} or of religious polarization on conflict~\cite{ABUBADER2019102248}).

Moreover, the decoupling of axioms and our macro–micro equivalences sheds light on the interpretation of the axioms. If one of the axioms does not fit the observed properties of antagonism, the practitioner may disregard that axiom while continuing to rely on the remaining ones. This departs from Esteban and Ray’s approach, in which all axioms were intertwined in proving the characterization theorem. For example, it has been shown that in some contexts the marginal effect of distance on alienation becomes negligible at large distances (e.g. \cite{10.1162/JEEA.2009.7.6.1291}). Using our equivalences, we can see that this assumption violates only Axiom 2 because, in this case, the alienation function would not be convex. In that case, one may remove Axiom 2, adopt an alienation function with a plateau beyond a threshold, and obtain a polarization index that still satisfies Condition H and Axioms 1 and 3.\\

The paper is organized as follows. Section \ref{model} reintroduces the model of Esteban and Ray. Section \ref{H} is devoted to study Condition H and its effect on antagonism functions. Section \ref{Axioms} provides necessary and sufficient conditions for each of Esteban and Ray’s three axioms in terms of antagonism. Section \ref{new} introduces new examples of polarization indices. Finally, Section \ref{expl} compares the performance of these new indices with the classical Esteban–Ray indices in specific distributions.

\section{The Model by Esteban and Ray}\label{model}

As mentioned in the introduction, Esteban and Ray's model builds on the individual perception of the effective antagonism that each member of a society feels toward others. We review the model, focusing on the domains in which polarization indices are defined, avoiding the assumption that individuals within each group are infinitely divisible, and adjusting the axioms accordingly.

\subsection{Distributions and polarization indices}

Assume that a society consists of a finite number of individuals, each characterized by a continuous numerical variable $Y$.\footnote{In ER, the characteristic is the logarithm of income. However, their indices have been applied to several other variables.} Let $\bm y=(y_1, y_2, \dots, y_n) \in \R^n$ denote the vector of observed values for the characteristic $Y$, and let $\bm \pi=(\pi_1, \pi_2, \ldots, \pi_n)$ denote the corresponding \emph{absolute numbers} of individuals with these characteristics, so that $N=\sum_i \pi_i$ is the total population. The absolute number of individuals in each group depends on the units used to measure individuals, for instance, $p=1, 10, 10^6, \ldots$ individuals. Define $u=\frac{1}{p}$ as the minimal value of $\pi_i$. We denote by $\N_u$ the set of (rational) numbers of the form $n u$, where $n\in \N$, and then $\bm \pi=(\pi_1, \pi_2, \ldots, \pi_n) \in {(\N_u)_{++}}^n$. It is crucial to note that $(\N_u)_{++}$ is bounded below by $u$, representing a single individual in the fixed units. An element $(\bm{\pi},\bm y)$ will be called a \emph{distribution}.

\begin{definition}\label{pol-index}
Consider a system of population units $u$. A \emph{polarization index} is a function $P: \mathcal D \to \R_+$, where 
\begin{equation*}
	\mathcal{D}=\bigcup_{n\ge2} {(\mathbb{N}_u)_{++}}^n \times \left\{(y_1,\ldots,y_n)\in \mathbb{R}^n: y_i\neq y_j \text{ for all distinct } i,j\in \{1,\ldots, n\}\right\}
\end{equation*}
is the space of all distributions for the fixed system of units $u$.
\end{definition}

For simplicity, we denote by $\mathcal{R}_n$ the set of observed value vectors corresponding to $n$ groups:
\[
\mathcal{R}_n=\left\{(y_1,\ldots,y_n)\in \mathbb{R}^n: y_i\neq y_j \text{ for all distinct } i,j\in \{1,\ldots, n\}\right\}.
\]

\begin{remark}
Esteban and Ray treat individuals within each group as infinitely divisible, that is, $\bm \pi=(\pi_1, \pi_2, \ldots, \pi_n) \in \R^n_{++}$. This assumption appears implausible in the present context. Since, as we will see later, the model relies on the antagonism that an individual $A$ with characteristic $y_A$ feels toward others whose characteristics differ from $y_A$, the group of individuals possessing characteristic $y_A$ must contain at least one individual $A$, implying $\pi_A\geq u$. 

An index $P$ defined for $(\pi,d)\in \R_{+}\times \R_+$ is also an index in the sense of Definition~\ref{pol-index}, simply by restriction. Hence, defining polarization indices on the domain used by ER serves as a tool to construct indices in the proper discrete domain. This domain extension, seemingly introduced by ER for technical convenience, is not innocuous, as we will later show.
\end{remark}

A basic principle that any polarization index should satisfy is the so-called \emph{Condition H}. Essentially, it requires that the ranking of two distributions in terms of polarization remain invariant under a homothecy with positive integral ratio of both population sizes.

\begin{definition}[Condition H]
A polarization index $P$ satisfies Condition H if, for any two distributions $(\bm{\pi},\bm y), (\bm{\pi}',\bm{y}') \in \mathcal D$ such that $P(\bm{\pi},\bm y)\geq P(\bm{\pi}',\bm{y}')$, it holds that
\[
P(\lambda\bm{\pi},\bm y)\geq P(\lambda\bm{\pi}',\bm{y}')
\]
for all $\lambda\in \N_{++}$.
\end{definition}

Our formulation differs slightly from that of Esteban and Ray, who assume the ratio to be a positive real number, that is $\lambda\in \R_{++}$. We restrict $\lambda$ to positive integers to ensure that $(\lambda\bm{\pi},\bm y)\in \mathcal{D}$, i.e., the resulting distribution is also admissible. Nevertheless, whenever the four distributions $(\bm{\pi},\bm y)$, $(\bm{\pi}',\bm{y}')$, $(\lambda\bm{\pi},\bm y)$, and $(\lambda\bm{\pi}',\bm{y}')$ all belong to $\mathcal D$, the implication
\[
P(\lambda\bm{\pi},\bm y)\geq P(\lambda\bm{\pi}',\bm{y}') \implies P(\bm{\pi},\bm y)\geq P(\bm{\pi}',\bm{y}')
\]
also holds.

\subsection{Polarization indices induced by effective antagonism}

We now define the \emph{effective antagonism} that an individual in the society feels toward any other. This is a function $\theta:\mathbb{N}_u\times \mathbb{R}_+\rightarrow \mathbb{R}$ depending on two variables, $(\pi,d)$. For each member of the society, $\pi$ denotes the number of individuals who share exactly the same value of $Y$, and  $d$ is the distance between two individuals (in the Euclidean sense). 

Given a function $\theta(\pi,d)$, the polarization of a society with distribution $(\bm{\pi},\bm y)$ is defined as the sum of all individual effective antagonisms. Thus, given an effective antagonism function $\theta$, we define
\begin{equation}\label{index}
	P(\bm{\pi},\bm y)=\sum_{i=1}^n\sum_{j=1}^n\pi_i\pi_j\theta(\pi_i,|y_i-y_j|), \quad (\bm{\pi},\bm y)\in {(\mathbb{N}_u)_{++}}^n \times \mathcal{R}_n,
\end{equation}
as the polarization index induced by $\theta$.

There are several conditions imposed on $\theta$. Recall that the goal of ER is to introduce an index that captures whether a society can be grouped into very different clusters but that each of these groups is formed by very similar individuals. For this reason, they argue that  \emph{intra-group homogeneity} accentuates polarization as well as it does \emph{inter-group heterogeneity}. This is reflected in the following definition.

\begin{definition}\label{eff-ant}
A function $\theta:\mathbb{N}_u\times \mathbb{R}_+\rightarrow \mathbb{R}$ is an \emph{effective antagonism function} if it satisfies:
\begin{itemize}
	\item $\theta(\pi,d)$ is non-decreasing in both $\pi$ and $d$.
	\item $\theta(\pi,d)$ is continuous in its second argument.
	\item $\theta(\pi,0)=0$ for all $\pi\in \mathbb{N}_u$.
	\item $\theta(\pi,d)>0$ whenever $\pi,d>0$.
\end{itemize}
\end{definition}

In their article, ER do not explicitly require $\theta(\pi,d)$ to be non-decreasing in $\pi$, although they impose this property for $d$. Their justification is that monotonicity in $\pi$ arises \emph{ex post} from their main theorem. However, we regard it as more natural to include this property in the definition itself\footnote{In fact, the argument used by ER to prove that $\theta(\pi,d)$ is increasing in $\pi$ is incorrect. They show that their identification function $\phi$ satisfies $\phi(2x)>\phi(x)$ for all $x$ and conclude that $\phi$ is increasing, which does not follow.}.

\begin{remark}
Although the proper domain of each polarization index is $\bigcup_{n\ge2} {(\mathbb{N}_u)_{++}}^n \times \mathcal{R}_n$, for simplicity we restrict attention here to the case $\pi_i\in \N$ (that is, $u=1$). The definitions and proofs for general $u$ are available under request.
\end{remark}

\section{The Structure of Antagonism Functions Compatible with Condition H}\label{H}

In the previous section, we introduced the model following Esteban and Ray. The only restrictions imposed thus far are Condition H and the basic properties that antagonism functions must satisfy. Remarkably, these assumptions alone---without any additional axioms---are sufficient to show that admissible antagonism functions have a very specific structure: $\theta(\pi,d)$ decomposes as the product of two functions, one depending only on $\pi$ and the other only on $d$. In other words, the effective antagonism can be expressed as the product of an \emph{identification function} $\phi(\pi)$ and an \emph{alienation function} $f(d)$. Moreover, we can establish that $\phi(\pi)=\pi^{\alpha}$. 

The outline of the proof is as follows. First, we show that $\theta(\pi,d)$ can be decomposed as $\theta(\pi,d)=\phi(\pi)f(d)$. Second, we prove that $\phi(\pi)$ is a completely multiplicative function, that is, $\phi(xy)=\phi(x)\phi(y)$. Recall that $\pi\in\N$. The structure of such functions may in general be complex,\footnote{This contrasts with the case $\pi\in\R$, where it is straightforward to verify that if $\phi(0)=0$ and $\phi$ is non-decreasing, then $\phi(\pi)=k\pi^{\alpha}$ for some $\alpha>0$.} since they must be defined on the prime numbers. However, in a classical paper by Erd\"os (1946)~\cite{Erdos1946}, it was shown that under the assumption of monotonicity these functions indeed take the desired form.

\begin{theorem}\label{H-increasing}
Let $\theta:\mathbb{N}\times \mathbb{R}_{+}\rightarrow \mathbb{R}$ be an antagonism function as defined in Definition~\ref{eff-ant}, and let $P$ be the corresponding polarization index. Then $P$ satisfies Condition H if and only if there exist a continuous and non-decreasing function $f:\mathbb{R}_{+}\rightarrow \mathbb{R}$ satisfying $f(0)=0$ and $f(d)>0$ for all $d>0$, and a fixed real number $\alpha\ge0$, such that for every $\pi\in\mathbb{N}$ and $d\in\mathbb{R}_{+}$, 
\[
\theta(\pi,d)=\pi^{\alpha}f(d).
\]
\end{theorem}

\begin{proof}
The first part of the proof shows that under the conditions of the theorem, $\theta(\pi,d)$ can be written as $\theta(\pi,d)=\phi(\pi)f(d)$, where $\phi(\pi)$ is completely multiplicative.

Assume that Condition H holds. We first prove that for every $\lambda\in\N_{++}$ there exists an exponent $\alpha(\lambda)$ such that, for all $\pi$ and $d$, 
\[
\theta(\lambda\pi,d)=\lambda^{\alpha(\lambda)}\theta(\pi,d).
\]
To establish this, fix $\lambda$ and consider $\pi,\pi'\in\N_{++}$ and $d,d'\in\R_{++}$. There exist $\alpha,\alpha'\in\R$ such that $\theta(\lambda\pi,d)=\lambda^{\alpha}\theta(\pi,d)$ and $\theta(\lambda\pi',d')=\lambda^{\alpha'}\theta(\pi',d')$.\footnote{Such $\alpha$ and $\alpha'$ exist because $\theta(\pi,d)>0$ whenever $\pi,d>0$, so $\alpha=\ln\!\big(\theta(\lambda\pi,d)/[\lambda\theta(\pi,d)]\big)$. When $d=0$ or $\pi=0$, there is nothing to prove since $\theta(\lambda\pi,d)=\lambda^{\alpha}\theta(\pi,d)=0$ for any $\alpha\in\R$.} We must then prove that $\alpha=\alpha'$, i.e., that $\alpha$ depends only on $\lambda$.

Let $(q_k)$ be a sequence of rational numbers converging to $\frac{\pi^2\theta(\pi,d)}{{\pi'}^2\theta(\pi',d')}$ such that $q_k>\frac{\pi^2\theta(\pi,d)}{{\pi'}^2\theta(\pi',d')}$ for all $k$. Each term can be written as $q_k=\frac{n_k}{m_k}$ with $n_k,m_k\in\N_{++}$. Since $\theta$ is nonnegative,
\[
n_k{\pi'}^2\theta(\pi',d')>m_k\pi^2\theta(\pi,d).
\]

Consider the $(m_k+1)$-tuples $\bm{\pi}=(\pi,\ldots,\pi)$ and $\bm{y}_{\epsilon}=(0,d,d+\epsilon,\ldots,d+(m_k-1)\epsilon)$, and the $(n_k+1)$-tuples $\bm{\pi}'=(\pi',\ldots,\pi')$ and $\bm{y}'_{\epsilon}=(0,d',d'+\epsilon,\ldots,d'+(n_k-1)\epsilon)$ for $\epsilon>0$ sufficiently small. Using the continuity of $\theta$ in $d$ and the fact that $\theta(\cdot,0)\equiv0$, we have
\[
\lim_{\epsilon\to0^+} P(\bm{\pi},\bm{y}_{\epsilon})=2m_k\pi^2\theta(\pi,d), \qquad 
\lim_{\epsilon\to0^+} P(\bm{\pi}',\bm{y}'_{\epsilon})=2n_k{\pi'}^2\theta(\pi',d').
\]
Given the strict inequality above, there exists $\epsilon_0>0$ such that $P(\bm{\pi},\bm{y}_{\epsilon})<P(\bm{\pi}',\bm{y}'_{\epsilon})$ for all $\epsilon<\epsilon_0$. Applying Condition H yields
\[
P(\lambda\bm{\pi},\bm{y}_{\epsilon})\le P(\lambda\bm{\pi}',\bm{y}'_{\epsilon}) \quad \text{for all } \epsilon<\epsilon_0.
\]
Taking limits gives
\[
2n_k(\lambda\pi')^2\lambda^{\alpha'}\theta(\pi',d')=\lim_{\epsilon\to0^+}P(\lambda\bm{\pi}',\bm{y}'_{\epsilon})
\ge \lim_{\epsilon\to0^+}P(\lambda\bm{\pi},\bm{y}_{\epsilon})
=2m_k(\lambda\pi)^2\lambda^{\alpha}\theta(\pi,d),
\]
which implies
\[
q_k=\frac{n_k}{m_k}\ge \lambda^{\alpha-\alpha'}\frac{\pi^2\theta(\pi,d)}{{\pi'}^2\theta(\pi',d')}.
\]
Since $(q_k)$ converges to $\frac{\pi^2\theta(\pi,d)}{{\pi'}^2\theta(\pi',d')}$, we conclude that $\lambda^{\alpha-\alpha'}\le1$. A symmetric argument with a sequence converging from below yields $\lambda^{\alpha-\alpha'}\ge1$, hence $\alpha=\alpha'$.

Thus, for each $\lambda$ there exists $\alpha(\lambda)$ such that $\theta(\lambda\pi,d)=\lambda^{\alpha(\lambda)}\theta(\pi,d)$ for all $\pi\in\N$ and $d\in\R_+$. Define $f(d)=\theta(1,d)$ and $\phi(\pi)=\pi^{\alpha(\pi)}$ for $\pi>0$ (and $\phi(0)=0$). It follows that
\[
\theta(\pi,d)=\pi^{\alpha(\pi)}\theta(1,d)=\phi(\pi)f(d).
\]
Moreover, $\phi$ is completely multiplicative, since
\[
\phi(\pi\pi')=(\pi\pi')^{\alpha(\pi\pi')}
=\frac{\theta(\pi\pi',1)}{\theta(1,1)}
=\pi^{\alpha(\pi)}\frac{\theta(\pi',1)}{\theta(1,1)}
=\phi(\pi)\phi(\pi').
\]

For the converse, suppose $\theta(\pi,d)=\phi(\pi)f(d)$ satisfies the hypothesis of the theorem. Let $\bm{\pi}=(\pi_1,\ldots,\pi_n)$, $\bm{y}=(y_1,\ldots,y_n)$, $\bm{\pi}'=(\pi'_1,\ldots,\pi'_m)$, and $\bm{y}'=(y'_1,\ldots,y'_m)$. If $P(\bm{\pi},\bm{y})\ge P(\bm{\pi}',\bm{y}')$, then
\[
\sum_{i,j}\pi_i\pi_j\theta(\pi_i,|y_i-y_j|)\ge \sum_{i,j}\pi'_i\pi'_j\theta(\pi'_i,|y'_i-y'_j|).
\]
For any $\lambda>0$,
\begin{equation*}
    P(\lambda\bm{\pi},\bm{y})=\sum_{i,j}\lambda^2\pi_i\pi_j \theta(\lambda \pi_i,|y_i-y_j|) = \lambda^2 \phi(\lambda) \sum_{i,j}\pi_i\pi_j \phi(\pi_i)f(|y_i-y_j|)\ge
\end{equation*}
\begin{equation*}
    \ge\lambda^2 \phi(\lambda) \sum_{i,j}\pi_i\pi_j \phi(\pi_i)f(|y_i-y_j|) = \sum_{i,j}\lambda^2\pi'_i\pi'_j \theta(\lambda \pi'_i,|y'_i-y'_j|) =  P(\lambda\bm{\pi'},\bm{y'})
\end{equation*}
so $P$ satisfies Condition H.

Hence, an effective antagonism function $\theta$ induces a polarization index satisfying Condition H if and only if $\theta(\pi,d)=\phi(\pi)f(d)$, where $\phi$ is completely multiplicative, that is, $\phi(xy)=\phi(x)\phi(y)$.

Finally, since $\theta$ is non-decreasing in $\pi$, $\phi$ must also be non-decreasing. Completely multiplicative, non-decreasing functions $\phi:\N\to\R$ have been characterized in \cite{Erdos1946,Nowacki1979}: they are precisely of the form $\phi(\pi)=\pi^{\alpha}$ for some $\alpha\ge0$, as desired.
\end{proof}

\begin{remark}
The proof extends to distributions defined on $\mathcal D=\bigcup_{n\ge2}\mathbb{R}_{++}^n\times\mathcal{R}_n$. Hence, our result also holds in the framework used by Esteban and Ray.
\end{remark}

\section{The Effect of Axioms 1--3 on the Allowed Antagonism Functions}\label{Axioms}

Esteban and Ray equipped their model with three additional axioms that describe how the polarization index behaves under specific transformations of certain distributions. Essentially, these axioms state that if one modifies a three-group distribution in such a way that it becomes ``closer'' to a two-bar bimodal distribution, then polarization must increase. 

In this section, we revisit these conditions, providing a mathematically precise formulation within a framework where individuals are not infinitely divisible. We then establish the necessary and sufficient conditions for each axiom, under the assumption that the effective antagonism must be of the form $\theta(\pi,d)=\pi^{\alpha}f(d)$ in order to satisfy Condition H.

\subsection{On Axiom 1}

Axiom 1 concerns the effect of merging two small groups into a single one when we have a three-group distribution.

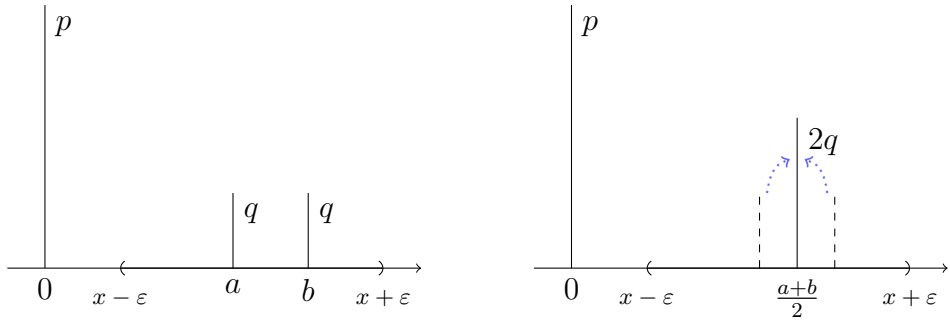
\begin{figure}[h!]
	\centering
\begin{tikzpicture}[scale=1]
	
	\newcommand\barra[4]{
		\draw (#1,0) -- (#1,#2);
		\node[below right] at (#1,#2) {$#4$};
		\node[below] at (#1,0) {$#3$};
	}
	
	\newcommand\barradisc[4]{
		\draw[dashed] (#1,0) -- (#1,#2);
		\node[below right] at (#1,#2) {$#4$};
		\node[below] at (#1,0) {$#3$};
	}
	
	\begin{scope}[shift={(-3.5,0)}, scale=1, every node/.style={scale=1}]
		
		\draw[->] (-0.5,0) -- (5,0) node[right] {};
		
		\barra{2.5}{1}{\scriptsize a}{q};
		\barra{3.5}{1}{b}{q};
		\barra{0}{3.5}{0}{p};
		
		\draw[{(-)}] (1,0) node[label=below:{\scriptsize$x - \varepsilon$}] {} -- (4.5,0) node[label=below:{\scriptsize$x + \varepsilon$}] {};
		
	\end{scope}
	
	\begin{scope}[shift={(3.5,0)}, scale=1, every node/.style={scale=1}]
		
		\draw[->] (-0.5,0) -- (5,0) node[right] {};
		
		\barradisc{2.5}{1}{}{};
		\barradisc{3.5}{1}{}{};
		
		\draw[->, blue!60, thick, dotted] (2.6,1) .. controls (2.65,1.4) and (2.8,1.4) .. (2.9,1.45);
		\draw[->, blue!60, thick, dotted] (3.4,1) .. controls (3.35,1.4) and (3.2,1.4) .. (3.1,1.45);
		
		\barra{3}{2}{\frac{a+b}{2}}{2q};
		\barra{0}{3.5}{0}{p};
		
		\draw[{(-)}] (1,0) node[label=below:{\scriptsize$x - \varepsilon$}] {} -- (4.5,0) node[label=below:{\scriptsize$x + \varepsilon$}] {};
		
	\end{scope}
	
\end{tikzpicture}
\caption{Graphical representation of Axiom 1. The distribution on the right is more polarized.}
\label{figure_Axiom1}
\end{figure}

The axiom can be summarized as follows. Suppose we have the three-group configuration on the left of Figure~\ref{figure_Axiom1}, consisting of a large group and two smaller groups of equal size. Suppose further that these two smaller groups are \emph{close enough} relative to their distance from the larger one. Then, if we merge the two smaller groups into a single new group located at the midpoint of the originals, polarization should increase. This will hold, at least, when the two smaller groups are \emph{small enough} relative to the larger one, ensuring that the new group is not excessively large—otherwise the resulting distribution may no longer resemble a two-bar bimodal shape.

Note that the original formalization of Axiom 1 by Esteban and Ray required a correction introduced by Kawada \emph{et al.}~\cite{KAWADA201835} to make their main theorem valid. We recall their revised version here.

\begin{axiom*}[Axiom 1 of Esteban and Ray (from \cite{KAWADA201835})]\label{axiom1ER}
For any $p > 0$ and $x > 0$, there exist $\epsilon > 0$ and $\mu > 0$ such that for any $a,b \in B(x,\epsilon)$ and any $q < p$ with $0 < q < \mu p$, one has
\[
P ((p, q, q), (0, a, b)) < P((p, 2q), (0, \frac{a + b}{2})).
\]
\end{axiom*}

While this captures the intuition behind Axiom 1, it only makes sense if individuals can be treated as infinitely divisible. The role of $\mu$ in this definition is precisely to ensure that for all values of $q$ not smaller but small enough with respect to $p$ (that is, $q < \mu p$), one has the desired effect in polarization. However for distributions $(\bm{\pi},\bm{y})\in \bigcup_{n\ge2} {\mathbb{N}_{++}}^n \times \mathcal{R}_n$, where the number of individuals in a group is bounded below, taking $\mu<1/p$ the axiom is empty and then it is satisfied trivially because there is no $q<\mu p$ contradicting the axiom. To preserve the spirit of Axiom 1 while making it meaningful, we reformulate it as follows.

\begin{axiom}[Axiom 1]\label{Axiom1-CR}
Let $p > 1$ and $x > 0$. There exist $\epsilon > 0$ and $q_0$ with $0 < q_0 < p$ such that if $q \leq q_0$ and $a,b \in B(x,\epsilon)$, then
\[
P ((p, q, q), (0, a, b)) < P((p, 2q), (0, \frac{a + b}{2})).
\]
\end{axiom}

The Gini index, corresponding to $\theta(\pi,d) = d$ (i.e., $\alpha = 0$ and $f(d) = d$), satisfies neither the original Axiom 1 of Esteban and Ray nor Axiom~\ref{Axiom1-CR}. We will now prove that Axiom 1 holds if and only if the identification factor in the effective antagonism is not neutral, in other words, if $\alpha \neq 0$ for any given alienation function $f$.

\begin{proposition}\label{axiom1-characterization}
Let $P$ be a polarization index induced by an antagonism function satisfying Condition H, that is, $\theta(\pi, d) = \pi^\alpha f(d)$ with $\alpha \geq 0$, $f(0) = 0$, and $f$ non-decreasing. Then $P$ satisfies Axiom 1 if and only if $\alpha > 0$.
\end{proposition}

\begin{proof}
Assume first that Axiom 1 holds and that $\theta(\pi, d) = \pi^\alpha f(d)$ with $f$ non-decreasing and $\alpha \geq 0$. Suppose $\alpha = 0$. Consider any $x > 0$, $p = 2$, and $q = 1$ (the smallest $q$ satisfying $q < p$). By Axiom 1, there exists $\epsilon > 0$ such that for every $a,b \in B(x,\epsilon)$ with $a < b$, the initial and final polarization levels $P_i$ and $P_f$ satisfy
\[
P_i = P((2,1,1),(0,a,b)) < P((2,2),(0, \frac{a+b}{2})) = P_f.
\]
Now,
\[
P_i = 4(f(a) + f(b)) + 2f(|b-a|), \qquad
P_f = 8f\left(\frac{a+b}{2}\right).
\]
Hence,
\begin{equation}\label{difference-axiom1}
0 < \frac{P_f - P_i}{4} = 2\left(f\left(\frac{a+b}{2}\right) - \frac{f(a) + f(b)}{2}\right) - \frac{f(|b-a|)}{2}.
\end{equation}
Since $f(d) \ge 0$ for all $d$, it follows that for every $x > 0$, there exists $\epsilon > 0$ such that for all $a,b \in B(x,\epsilon)$ with $a < b$,
\[
0 < f\left(\frac{a+b}{2}\right) - \frac{f(a) + f(b)}{2}.
\]
That is, $f$ is locally strictly midpoint concave. Since every continuous midpoint-concave function is also concave (see, e.g., \cite[Appendix~C]{KAWADA201835}), $f$ is concave on $(0,\infty)$ and, by continuity, also on $[0,\infty)$.

Returning to Inequation~\eqref{difference-axiom1} and rearranging terms becomes
\begin{equation}\label{ineqA1}
0 < \left(f\left(\frac{a+b}{2}\right) - f(b)\right) + \left(f\left(\frac{a+b}{2}\right) - f(a)\right) - \frac{f(|b-a|)}{2}.
\end{equation}
Using a well known result on the decrease of the slope of the secants of concave functions and that $f$ is  continuous (see  for instance \cite[Theorem 7.5]{Sundaram_1996}) one has:
\[
\frac{f(b-a)}{b-a} \ge \frac{f(a) - f(b-a)}{b} \ge \frac{f\left(\frac{a+b}{2}\right) - f(a)}{\frac{b-a}{2}}.
\]
Then,
\[
f\left(\frac{a+b}{2}\right) - f(a) - \frac{1}{2}f(|b-a|) \le 0,
\]
and since $f$ is non-decreasing,
\[
0 < \left(f\left(\frac{a+b}{2}\right) - f(b)\right) + \left(f\left(\frac{a+b}{2}\right) - f(a)\right) - \frac{f(|b-a|)}{2} \le 0,
\]
a contradiction.\\

For the converse, assume $\theta(\pi, d) = \pi^\alpha f(d)$ with $\alpha > 0$ and $f$ non-decreasing. For any $x > 0$, let $P_f$ and $P_i$ denote, respectively, the polarization levels after and before the merger described in Axiom 1, for arbitrary $q$, $a$, and $b$ such that $q < p$. Then
\[
P_f = P\left((p,2q), \left(0,\frac{a+b}{2}\right)\right) = 2pq(p^\alpha + (2q)^\alpha)f\left(\frac{a+b}{2}\right),
\]
and
\[
P_i = P((p,q,q),(0,a,b)) = pq(p^\alpha + q^\alpha)(f(a)+f(b)) + q^{\alpha+2}f(|b-a|).
\]
Since $f$ is continuous, $\alpha > 0$, and $f(x) > 0$,
\[
\lim_{a\to x}\lim_{b\to x} (P_f - P_i) = 2pq^{\alpha+1}(2^\alpha - 1)f(x) > 0.
\]
Hence, there exists a neighborhood $B(x,\epsilon_q)$ of $x$, depending on $q$, such that if $a,b \in B(x,\epsilon_q)$, then $P_f > P_i$. As there are finitely many $q < p$ with $q \in \mathbb{N}$, we may set $\epsilon := \min_{q < p} \epsilon_q > 0$, concluding that for all $q < p$ and all $a,b \in B(x,\epsilon)$, $P_f > P_i$.
\end{proof}

\subsection{On Axiom 2}

The second axiom introduced by Esteban and Ray concerns the effect of changes that bring groups in a distribution closer together.

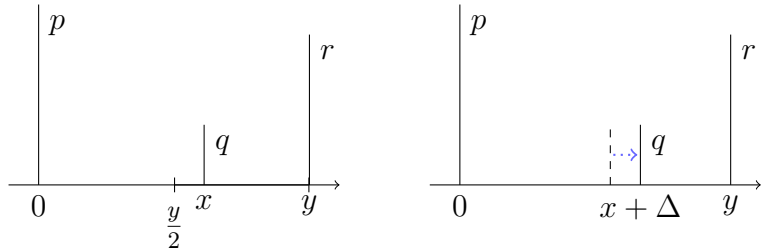
\begin{figure}[h!]
	\centering
\begin{tikzpicture}[scale=0.8]
	
	\newcommand\barra[4]{
		\draw (#1,0) -- (#1,#2);
		\node[below right] at (#1,#2) {$#4$};
		\node[below] at (#1,0) {$#3$};
	}
	
	\newcommand\barradisc[4]{
		\draw[dashed] (#1,0) -- (#1,#2);
		\node[below right] at (#1,#2) {$#4$};
		\node[below] at (#1,0) {$#3$};
	}
	
	\begin{scope}[shift={(-3.5,0)}
		, scale=1, every node/.style={scale=1}
		]
		
		\draw[->] (-0.5,0) -- (5,0) node[right] {};

		\barra{2.75}{1}{x}{q};
		\barra{4.5}{2.5}{y}{r};
		\barra{0}{3}{0}{p}
		\draw[{|-|}] (2.25,0) node[label=below:{\scriptsize$\displaystyle \frac{y}{2}$}] {} -- (4.5,0) node[label=below:{}] {};
		
	\end{scope}
	
	\begin{scope}[shift={(3.5,0)}
		, scale=1, every node/.style={scale=1}
		]
		
		\draw[->] (-0.5,0) -- (5,0) node[right] {};
		
		\barradisc{2.5}{1}{}{};
		\barra{3}{1}{x+\Delta}{q};
		\barra{4.5}{2.5}{y}{r};
		\barra{0}{3}{0}{p}
		
		\draw[->, blue!60, thick, dotted] (2.55,0.5) -- (2.95,0.5);
		
	\end{scope}

\end{tikzpicture}

\caption{Graphical representation of Axiom 2. The distribution on the right is more polarized.}
\label{figure_Axiom2}
\end{figure}

Roughly speaking, Axiom 2 can be described as follows: Consider the three-group configuration on the left of Figure~\ref{figure_Axiom2}. Assume that the two extreme groups are unequal and that the central group is closer to the smaller extreme group. Then, as the central group moves toward the midpoint between the two extreme groups, the overall polarization decreases.

\begin{axiom}[Axiom 2 (from \cite{KAWADA201835})]
	For any $p, q, r > 0$ with $p > r$, any $x, y > 0$ such that $|y-x| < x < y$, and any $\Delta \in (0, y - x)$,
	\[
	P((p, q, r), (0, x, y)) < P((p, q, r), (0, x + \Delta, y)).
	\]
\end{axiom}

Next, we show that for polarization indices satisfying Condition H, Axiom 2 is equivalent to requiring that the alienation function $f$ be convex.

\begin{proposition}\label{axiom2-characteriation}
Let $P$ be a polarization index induced by an antagonism function satisfying Condition H, that is, $\theta(\pi, d)=\pi^\alpha f(d)$ with $\alpha \geq 0$, $f(0)=0$, and $f$ non-decreasing. Then $P$ satisfies Axiom 2 if and only if $f$ is convex.
\end{proposition}
\begin{proof}
Assume first that $\theta(\pi,d)=\pi^\alpha f(d)$ with $f$ convex. The polarization associated with Axiom 2 is given by
\begin{equation*}
	\mathcal{P}(\Delta)
	= pq(p^{\alpha}+q^{\alpha}) f(x+\Delta)
	+ pr(p^{\alpha}+r^{\alpha})f(y)
	+ qr(q^{\alpha}+r^{\alpha})f(y-x-\Delta),
\end{equation*}
with $0 \le \Delta < y-x$. In particular, $\mathcal{P}(0)$ represents the polarization of the initial distribution. To verify Axiom 2, we must show that $\mathcal{P}(\Delta)-\mathcal{P}(0) > 0$.
\begin{equation*}
	\mathcal{P}(\Delta)-\mathcal{P}(0)
	= pq(p^{\alpha}+q^{\alpha})(f(x+\Delta)-f(x))
	+ rq(q^{\alpha}+r^{\alpha})(f(y-x-\Delta)-f(y-x)).
\end{equation*}
Let $u=pq^{-1}$ and $v=rq^{-1}$, and define $z_0=y-x-\Delta$ and $z_1=x+\Delta$. Observe that $z_0=y-x-\Delta < y-x < x < x+\Delta= z_1$. Hence, there exist $t,\bar{t}\in(0,1)$ such that $x=tz_0+(1-t)z_1$ and $y-x=\bar{t}z_0+(1-\bar{t})z_1$. Moreover, since the distances between these points satisfy $d(z_0, y-x)=d(z_1,x)$, it follows that $t+\bar{t}=1$. Then
\begin{equation*}
	\frac{\mathcal{P}(\Delta)-\mathcal{P}(0)}{q^{\alpha+2}}
	= u(u^{\alpha}+1)(f(z_1)-f(x))
	+ v(v^{\alpha}+1)(f(z_0)-f(y-x)).
\end{equation*}

Because $u>v$, we obtain
\begin{equation*}
	\frac{\mathcal{P}(\Delta)-\mathcal{P}(0)}{q^{\alpha+2}}
	> v(v^{\alpha}+1)\big[(f(z_1)-f(x))+(f(z_0)-f(y-x))\big].
\end{equation*}

Applying the convexity of $f$, that is, $f(x)\le tf(z_0)+(1-t)f(z_1)$ and $f(y-x)\le \bar{t}f(z_0)+(1-\bar{t})f(z_1)$, yields
\begin{equation*}
		\frac{\mathcal{P}(\Delta)-\mathcal{P}(0)}{q^{\alpha+2}}>v(v^\alpha+1)\left[\left( f(z_1)-tf(z_0)-(1-t)f(z_1)\right)+\left(f(z_0)-\bar{t}f(z_0)-(1-\bar{t})f(z_1)\right)\right]=
	\end{equation*}
	\begin{equation*}
		=v(v^\alpha+1)\left(f(z_1)-f(z_0)\right)(t+\bar{t}-1)=0
	\end{equation*}
since $t+\bar{t}=1$. Hence $\mathcal{P}(\Delta)-\mathcal{P}(0)>0$.\\

For the converse, assume that Axiom 2 holds. We must show that $f$ is convex. Let $z_0,z_1\in \mathbb{R}$ with $z_0<z_1$. First, we will prove that
\[
f\!\left(\frac{z_0+z_1}{2}\right)\le\frac{f(z_0)+f(z_1)}{2}.
\]
Define $x=\tfrac{z_0+z_1}{2}$, $y_n=2x-\tfrac{z_0}{n+1}$, and $\Delta=\tfrac{z_1-z_0}{2}$. Let $(p_n,q_n,r_n)$ be given by $p_n=n+2$, $q_n=n$, and $r_n=n+1$. Then $u_n=p_n{q_n}^{-1}$ and $v_n=r_n{q_n}^{-1}$ both converge to $1$.

For every $n>0$, the vectors $(p_n,q_n,r_n)$ and $(0,x,y_n)$ together with $\Delta$ satisfy the assumptions of Axiom 2. Let
\[
\mathcal{P}_{n,\Delta}=P((p_n,q_n,r_n),(0,x+\Delta,y_n)),\qquad
\mathcal{P}_{n,0}=P((p_n,q_n,r_n),(0,x,y_n)).
\]
Then
\begin{equation*}
	0<
	\frac{\mathcal{P}_{n,\Delta}-\mathcal{P}_{n,0}}{q_n^{\alpha+2}}
	= u_n(u_n^{\alpha}+1)(f(x+\Delta)-f(x))
	+ v_n(v_n^{\alpha}+1)(f(y_n-x-\Delta)-f(y_n-x)).
\end{equation*}

Taking the limit as $n\to\infty$ and using the continuity of $f$, we obtain
\[
0 \le 2\left(f(z_1)-f\!\left(\frac{z_0+z_1}{2}\right) + f(z_0)-f\!\left(\frac{z_0+z_1}{2}\right)\right),
\]
which implies
\[
f\!\left(\frac{z_0+z_1}{2}\right)\le\frac{f(z_0)+f(z_1)}{2}.
\]
Therefore, $f$ is midpoint convex on $(0,\infty)$, and since $f$ is continuous, it is convex on $\mathbb{R}_+$.
\end{proof}

\subsection{On Axiom 3}
The last axiom proposed by ER concerns the effect of transferring individuals from a central group to groups located at the extremes.

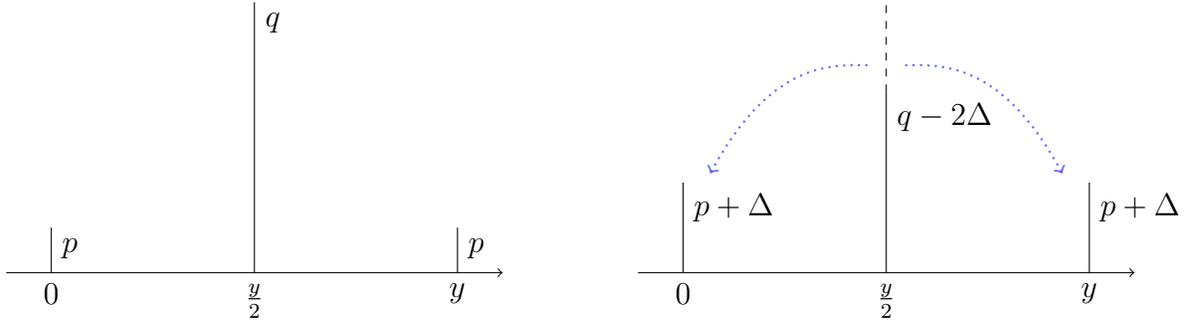
\begin{figure}[h]
	\centering
\begin{tikzpicture}[scale=1.2]

	\newcommand\barra[4]{
		\draw (#1,0) -- (#1,#2);
		\node[below right] at (#1,#2) {$#4$};
		\node[below] at (#1,0) {$#3$};
	}
	
	\newcommand\barradisc[4]{
		\draw[dashed] (#1,0) -- (#1,#2);
		\node[below right] at (#1,#2) {$#4$};
		\node[below] at (#1,0) {$#3$};
	}
	
	\begin{scope}[shift={(-3.5,0)}
		, scale=1, every node/.style={scale=1}
		]
		
		\draw[->] (-0.5,0) -- (5,0) node[right] {};

		\barra{4.5}{0.5}{y}{p};
		\barra{2.25}{3}{\frac{y}{2}}{q};
		\barra{0}{0.5}{0}{p};
		
	\end{scope}
	
	\begin{scope}[shift={(3.5,0)}
		, scale=1, every node/.style={scale=1}
		]
		
		\draw[->] (-0.5,0) -- (5,0) node[right] {};

		\barra{4.5}{1}{y}{p+\Delta};
		\barra{0}{1}{0}{p+\Delta};
		\barra{2.25}{2}{\frac{y}{2}}{q-2\Delta};
		\draw[dashed] (2.25,2) -- (2.25,3);
		
		\draw[<-, blue!60, thick, dotted] (0.3,1.1) .. controls (1,2.4) and (1.7,2.3) .. (2.1,2.3);
		\draw[<-, blue!60, thick, dotted] (4.2,1.1) .. controls (3.5,2.4) and (2.8,2.3) .. (2.4,2.3);

	\end{scope}

\end{tikzpicture}
	\caption{Graphical representation of Axiom 3}	
\label{figureAxiom3}
\end{figure}

Axiom 3 states that in a three–group distribution where two equal groups are located at the extremes and a third group lies at the midpoint, a transfer of individuals from the central group---half of them to each of the extreme groups---increases polarization. More precisely:

\begin{axiom}[Axiom 3]
	For any $p, q > 0$, and any $x, y > 0$ with $x = y - x$, and any
	$\Delta \in (0, q/2)\cap \mathbb{N}$,
	\[
	P((p, q, p), (0, x, y)) < P((p + \Delta, q - 2\Delta, p + \Delta), (0, x, y)).
	\]
\end{axiom}

As in ER, Axiom 3 imposes bounds on the admissible values of $\alpha$ in the identification factor, depending on the alienation function $f(d)$. Let
\[
g(p,q,\alpha)
= \frac{p^{\alpha+1} q + p q^{\alpha+1} - (p+1)^{\alpha+1} (q-2) - (p+1) (q-2)^{\alpha+1}}
{(p+1)^{\alpha+2} - p^{\alpha+2}}.
\label{A3-alpha}
\]

\begin{proposition}\label{Axiom3-Characterization}
Let $P$ be a polarization index induced by an antagonism function $\theta(\pi,d)=\pi^\alpha f(d)$ with $\alpha \geq 0$, $f(0)=0$, and $f$ non-decreasing. Then $P$ satisfies Axiom 3 if and only if, for every $d \in \mathbb{R}_{++}$ and every $p,q \in \mathbb{N}$ with $p>0$ and $q \ge 2$,
\begin{equation} \label{inequation:Axiom3}
	\frac{f(2d)}{f(d)} > g(p,q,\alpha).
\end{equation}
\end{proposition}

\begin{proof}
Let $(p,q,p)$ be a population vector with $q \ge 2$, $p>0$, and $d>0$, as in the statement of Axiom 3. The initial (left-hand side of Figure~\ref{figureAxiom3}) and final (right-hand side) polarizations after transferring exactly two individuals are:
\[
\begin{aligned}
	P_i &= 2 \big[pq (p^{\alpha} + q^{\alpha}) f(d) + p^{\alpha+2} f(2d)\big], \\
	P_f &= 2 \big[(p+1)(q-2)((p+1)^{\alpha} + (q-2)^{\alpha})f(d) + (p+1)^{\alpha+2} f(2d)\big].
\end{aligned}
\]
Thus, the polarization measure induced by $\theta$ satisfies Axiom 3 if and only if $P_f - P_i > 0$ for all $p>0$ and $q \ge 2$.

This inequality is equivalent to
\[
\frac{f(2d)}{f(d)} >
\frac{p^{\alpha+1} q + p q^{\alpha+1} - (p+1)^{\alpha+1} (q-2) - (p+1) (q-2)^{\alpha+1}}
{(p+1)^{\alpha+2} - p^{\alpha+2}},
\]
which proves the claim.
\end{proof}

Although we have shown that Axiom 1 concerns only identification and Axiom 2 concerns only alienation, Proposition~\ref{Axiom3-Characterization} shows that Axiom 3 connects the two components of antagonism functions. The left-hand side of inequality~\eqref{inequation:Axiom3} depends on the alienation function $f$, whereas the right-hand side depends solely on $\alpha$, the parameter defining the identification function. It is not immediately clear when inequality~\eqref{inequation:Axiom3} holds. However, in Section~\ref{new} we will show that the set of identification and alienation functions for which it does is quite large, leading to new polarization indices not considered by ER.

\section{New Polarization Indices that Satisfy the ER Model}\label{new}

In the previous two sections, we provided necessary and sufficient conditions for antagonism functions that induce indices satisfying Condition H and each of ER’s axioms. The combination of Condition H and Axioms~1 and~2 is equivalent to requiring that $\theta(\pi,d)=\pi^\alpha f(d)$ with $\alpha>0$ and $f$ a non-decreasing, convex function satisfying $f(0)=0$. However, it is not clear for which values of $\alpha$ the necessary and sufficient condition for Axiom 3, given in Proposition~\ref{Axiom3-Characterization}, holds depending on the choice of $f$, and conversely. The aim of this section is to shed light on this question.\\

We begin with the case $\alpha=1$, which is of particular interest in the literature. It is straightforward to verify that $\sup_{p>0,\,q\ge2} g(p,q,1)=1$ and that there are no $p,q$ satisfying the hypotheses of Proposition~\ref{Axiom3-Characterization} such that $g(p,q,1)=1$. This implies that an antagonism function of the form $\theta(\pi,d)=\pi f(d)$ satisfies Axiom 3 if and only if
\[
\frac{f(2d)}{f(d)} \ge 1.
\]
This inequality is trivially satisfied for any non-decreasing function $f(d)$. Hence, we obtain the following examples of antagonism functions that satisfy all the axioms proposed by ER.

\begin{example}
The following antagonism functions induce polarization indices satisfying Condition H and Axioms~1, 2, and~3:
\begin{itemize}
    \item $\theta(\pi, d)=\pi d$;
    \item $\theta(\pi, d)=\pi d^r$ with $r\ge 1$;
    \item $\theta(\pi, d)=\pi(\mathrm{e}^d-1)$;
    \item $\theta(\pi,d)=\pi f(d)$, where $f(d)$ is any continuous, convex, and non-decreasing function satisfying $f(0)=0$.
\end{itemize}
\end{example}

These examples already show that, for $\alpha=1$, there exist many antagonism functions capturing the idea of polarization beyond the single one identified by ER.

Next, we examine whether, as in ER, one can find an interval for $\alpha$ such that Axiom 3 is satisfied for any antagonism function of the form $\theta(\pi,d)=\pi^\alpha f(d)$ with $\alpha>0$ and $f(d)$ continuous, convex, and non-decreasing (that is, satisfying Axioms~1 and~2). We show that the interval found by ER, namely $(0,\alpha^*]$ with $\alpha^*\approx1.6$, is in fact the universal interval for which any antagonism function satisfying Condition H and Axioms~1 and~2 also satisfies Axiom 3.

\begin{proposition}\label{Axiom3-minimal-interval}
Let $\theta(\pi,d)=\pi^\alpha f(d)$ with $f(d)$ convex, non-decreasing, and satisfying $f(0)=0$. Then, if $\alpha\in(0,\alpha^*]$, $\theta$ induces a polarization index that satisfies Condition H and Axioms~1, 2, and~3.
\end{proposition}

\begin{proof}[Sketch of proof]
Since $f$ is convex and $f(0)=0$, using again the decrease of the slope of the secants \cite[Theorem~7.5]{Sundaram_1996} we have, for every $d>0$,
\[
\frac{f(2d)}{f(d)} \ge 2.
\]
Hence, if for a given $\alpha$ it holds that $M_\alpha:=\sup_{p>0,\,q\ge2} g(p,q,\alpha)<2$, then any index induced by $\theta(\pi,d)=\pi^\alpha f(d)$ with $f(d)$ continuous, convex, non-decreasing, and $f(0)=0$ will satisfy Proposition~\ref{Axiom3-Characterization}. Numerical computations show that $M_\alpha<2$ if and only if $\alpha\in(0,\alpha^*]$, where $\alpha^*\approx1.6$. These bounds coincide with those computed by ER for the specific case $\theta(\pi,d)=\pi^\alpha d$.
\end{proof}

As an illustration of Proposition~\ref{Axiom3-minimal-interval}, we have the following:

\begin{example}
For $\alpha\in(0,\alpha^*]$ with $\alpha^*\approx1.6$, the following antagonism functions induce polarization indices that satisfy Condition H and Axioms~1, 2, and~3:
\begin{itemize}
    \item $\theta(\pi, d)=\pi^\alpha d^r$ with $r\ge1$;
    \item $\theta(\pi, d)=\pi^\alpha P(d)$, where $P$ is any polynomial with non-negative coefficients and zero constant term;
    \item $\theta(\pi, d)=\pi^\alpha(\mathrm{e}^{kd}-1)$ with $k>0$.
\end{itemize}
\end{example}

Proposition~\ref{Axiom3-minimal-interval} therefore shows that the interval found by ER, $\alpha\in(0,\alpha^*]$ with $\alpha^*\approx1.6$, is universal for antagonism functions whose induced indices satisfy Axioms~1–3. However, for certain explicit alienation functions, the feasible values of $\alpha$ can exceed $\alpha^*$. For instance, if $f(d)=d^2$, then the necessary and sufficient condition in Proposition~\ref{Axiom3-Characterization} becomes
\[
\frac{f(2d)}{f(d)} = 4 > g(p,q,\alpha) \quad \text{for every } p>0 \text{ and } q\ge2.
\]
Numerical computations show that this condition holds for $\alpha \in [0,\alpha^{**}]$, where $\alpha^{**}\approx1.9$. Note that in the special case $\alpha=0$, the antagonism function $\theta(\pi,d)=d^2$ satisfies Axioms~2 and~3 but not Axiom 1.\footnote{A detailed study of the computation of feasible $\alpha$-intervals is available under request.}

\begin{remark}
The new indices obtained are of the form $\theta(\pi,d)=\pi^\alpha f(d)$, where $f$ is a continuous, non-decreasing, convex transformation of the distance. One could think that these new indices are equivalent to those obtained by ER with the Euclidean distance replaced by $f(d)$. However, $f(d)$ is generally not a distance, as it fails to satisfy the triangular inequality except when $f$ is linear.

A related question arises: do these new indices coincide with the ER index applied to a suitable transformation $g(y)$ of the underlying characteristic variable $Y$—for instance, $g(y)=\log(y)$, as in ER? The answer, again, is negative, except in the case where $f$ is linear.
\end{remark}

\section{Comparing the Performance of ER Indices and the New Indices}\label{expl}

\begin{figure}[ht]
	\centering
	\includegraphics[width=0.45\textwidth]{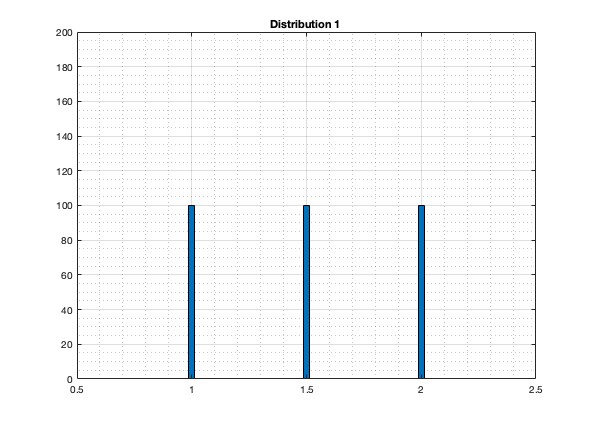}
	\includegraphics[width=0.45\textwidth]{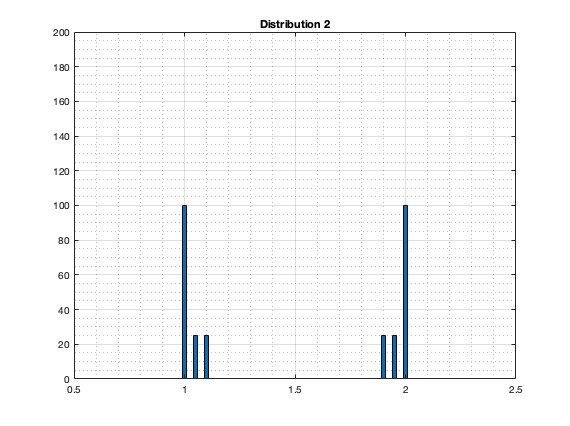}
	\caption{Two distributions with the same number of individuals.}	
	\label{Example1}
\end{figure}

Consider the two distributions in Figure~\ref{Example1}.  
Which of these is more polarized?  
A reader familiar with the ER paper will recognize them as a slight variation of the setting in Axiom 3.  
Both distributions contain the same number of individuals, and Distribution~2 is obtained by symmetrically transferring individuals initially in the central group to positions \emph{very close}\footnote{The maximum distance selected is $0.1$, merely to make the figures legible.} to the extreme groups.  
Following the argument in ER—namely, that \emph{``the disappearance of a middle class into the rich and poor categories must increase polarization''} \cite[p.~833]{10.2307/2951734}—one would expect Distribution~2 to be more polarized.  
However, using the classical ER index, Distribution~1 is actually more polarized than Distribution~2 for $\alpha \ge 0.3$.  
This may come as a surprise to practitioners who use these indices.  
Only when $\alpha$ is very small, and hence the ER index approaches the Gini index, does the distribution on the right become more polarized.  
In particular, for the standard value $\alpha = 1$, used in many empirical papers, this  does not occur.

By contrast, if one uses the antagonism function\footnote{The same result holds for $f(d) = d^n$, $n \ge 2$, or $f(d) = e^{k d} - 1$ with $k \ge 1$.} $\theta(\pi, d) = \pi^\alpha d^2$, the induced polarization index classifies Distribution~2 as more polarized than Distribution~1 for all $\alpha \in (0, \alpha^*]$, the universal interval for alienation functions satisfying all the axioms.

To illustrate that this is not a special case involving only three equal groups, consider the following example:

\begin{figure}[ht]
	\centering
	\includegraphics*[width=0.45\textwidth]{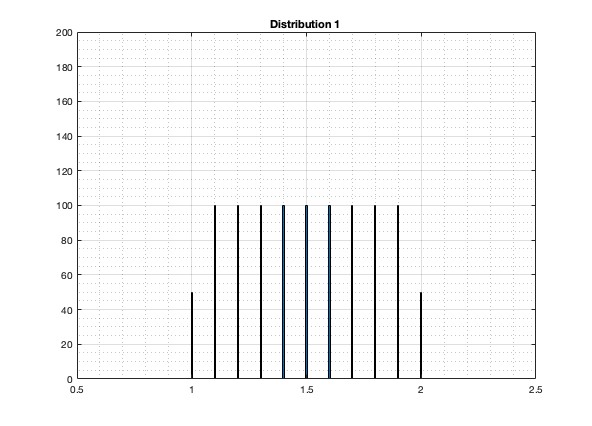}
	\includegraphics[width=0.45\textwidth]{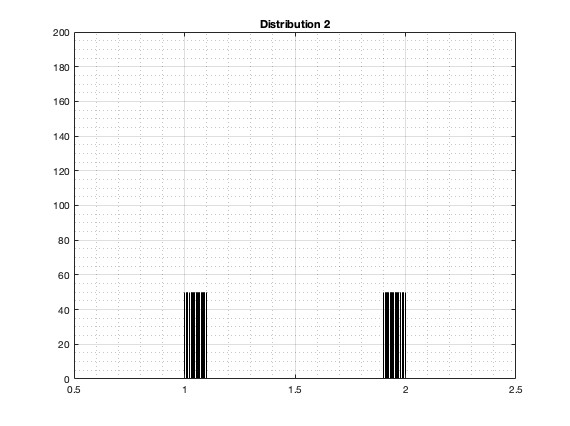}
\end{figure}

Both distributions contain the same number of individuals, but in Distribution~2 all individuals from the groups of $100$ members are now concentrated around $1$ and $2$, with each of these new clusters containing $45$ individuals.  
Following ER’s introduction, one might again think that Distribution~2 should be more polarized.  
Indeed, there are two clear clusters around $1$ and $2$, with high intra-group homogeneity (the maximum within-cluster distance is $md = 0.2$), and higher inter-group heterogeneity than in Distribution~1.  
Moreover, the ``middle class'' has disappeared.  
Nevertheless, this is not the case when using the ER indices for $\tilde{\alpha} > 0.5$.  
That is, except when, in ER’s words, the ``polarization sensitivity'' is very small and far from the standard value $\alpha = 1$.  
In fact, as the intra-group distance $md$ becomes very small, the critical value $\tilde{\alpha}$ stabilizes around $0.66$.

In contrast, the polarization index induced by $\theta(\pi, d) = \pi^\alpha d^2$ classifies Distribution~2 as more polarized than Distribution~1 for $\tilde{\alpha} < 1.45$, that is, except when the identification factor is extremely strong.  
Moreover, as $md$ decreases, this threshold $\tilde{\alpha}$ disappears, and Distribution~2 is classified as more polarized for all $\alpha \in (0, \alpha^*]$.

These are not exceptional cases.  
It is easy to construct further examples in which the orderings induced by ER appear counterintuitive.  
In our view, the constant marginal effect of distance in the ER index cannot compensate for the increase in the number of groups in such cases.

Practitioners should be aware of these, in our opinion, undesirable features of the classical ER indices.  
They must decide whether adopting an alienation function with a constant marginal effect is consistent with their beliefs about how antagonism operates in the phenomenon under study.

\section{Concluding Remarks}

In 1994, Esteban and Ray presented their model for measuring polarization based on effective antagonism functions. Despite its a priori appeal and strong axiomatic foundation, their characterization theorem \emph{``dramatically narrows the class of allowable polarization measures''}. This result rests heavily on a seemingly innocuous assumption: that individuals can be infinitely divisible. By removing this unnecessary assumption, we have shown that, within the original model and axiomatization, the concept of polarization admits a broader family of measures. We have also provided examples in which the original ER indices yield counter-intuitive classifications of societies, while most of the new indices we propose resolve these inconsistencies. A natural future work would be to find new, equally reasonable, axioms that restrict the set of polarization indices compatible with them in order to avoid these counter-intuitive effects.

Furthermore, we have derived necessary and sufficient conditions for antagonism functions satisfying Condition H, and, subsequently, for each of the axioms in ER in terms of identification and alienation. This analysis clarifies the precise effect of each axiom on the functional form of the effective antagonism function. In this way, our results translate the axiomatic structure proposed by ER from its macro formulation (in terms of distributional properties) into a micro formulation (in terms of the properties of antagonism) in which the model was originally conceived.

Also for future research, we suggest revisiting the continuous version of the ER model developed by Duclos, Esteban, and Ray~\cite{2990f903-a1fe-36dd-8511-1b59b65a89d2}. In that framework, the authors replace the unnormalized discrete distribution of individuals across finitely many groups with an unnormalized density function, so that the integral of the density over its support equals the total number of individuals. Density functions are introduced there to incorporate uncertainty in measuring the characteristic values across a population. However, a close reading of the proof of their main result, \cite[Theorem~1]{2990f903-a1fe-36dd-8511-1b59b65a89d2}, reveals that, as in ER, the argument relies crucially on the assumption that the total number of individuals in a well-defined subgroup can be made arbitrarily small. In their terminology, the total number of individuals in a \emph{basic density} can be infinitesimal.\footnote{This appears to contradict their statement in the introduction: \emph{``we do not mean to suggest that instances in which a single isolated individual runs amok with a machine gun are rare, or that they are unimportant in the larger scheme of things. It is just that these are not the objects of our enquiry.''}} As in the discrete case, this seems to be a simplifying assumption. We conjecture that, once it is relaxed, new polarization indices will emerge that also satisfy the axioms of the continuous model.

\nocite{*}
\bibliographystyle{plainnat}
\bibliography{bibliography}
\end{document}